\documentclass{ifacconf}

\usepackage{graphicx}      
\usepackage{natbib}        
\usepackage{bbm}
\usepackage{amsmath}
\usepackage{svg}
\usepackage{siunitx}
\usepackage{enumerate}
\DeclareMathOperator*{\argmin}{argmin}
\begin{document}
\begin{frontmatter}

\title{Distributed Synthesis of Gray-Box Distributed \(\mathcal{H}_2\) Controllers
  } 


\author[First]{Michael C. A. Nestor} and
\author[First]{Fei Teng} 

\address[First]{Department of Electrical and Electronic Engineering, Imperial College London, London, SW7 2AZ, UK 
    (e-mails: m.nestor22@imperial.ac.uk, f.teng@imperial.ac.uk).}

\begin{abstract}                
    Distributed controller synthesis offers scalable and privacy-preserving control design, but typical state-of-the-art approaches either assume white-box models or resort to centralized synthesis. In this paper, we combine partially known model knowledge and an input-state dataset within a distributed gray-box scheme to design \(\mathcal{H}_2\) controllers. Our method can handle unknown dynamics and offers scalable synthesis. Each agent communicates with a set of neighbors determined by the physical coupling topology of the system such that we can apply the Alternating Direction Method of Multipliers (ADMM) to solve the problem iteratively in a fully distributed fashion (i.e., without a central server). The effectiveness and flexibility of the proposed approach is demonstrated in simulations of the IEEE 39-bus power system test case.
\end{abstract}

\begin{keyword}
Cyber-physical systems, data-driven robust control, distributed robust controller synthesis, interconnected dynamical systems, multi-agent systems
\end{keyword}

\end{frontmatter}

\section{Introduction}

The real-time control of interconnected large-scale systems via a centralized controller poses operational challenges, as a single agent must handle all control communication traffic and computational effort between controller executions. Distributed control offers a scalable and privacy-preserving alternative for these systems, where multiple control agents act independently whilst sharing information in a peer-to-peer manner to enable coordination and guarantee stability. However, scalability challenges persist if the distributed controller itself is computed by solving a single system-wide structured control design problem (\cite{Rotkowitz-Lall-2006-Characterization-Convex-Problems-Decentralized-Control, Dvijotham-2015-Convex-Structured-Controller-Design-Finite-Horizon}). Distributed synthesis is an attractive way to achieve scalability by decomposing the global design problem into neighborhood sub-problems (\cite{De-Pasquale-2020-Extended-Full-Block-S-Procedure-Distributed-Control-Interconnected-Systems}).

Model-based controller design requires accurate knowledge of the system dynamics. However, such a model may be unavailable or infeasibly expensive to obtain, particularly in large-scale systems. Alternatively, controllers can be designed directly from data, avoiding an intermediate model identification stage (\cite{De-Persis-2020-Formulas-Data-Driven-Control}, \cite{Coulson-2019-Shallows-of-the-DeePC}). In realistic scenarios, partial model knowledge may be available leading to \emph{gray-box} methods; for example, an agent may know their local model but not the model of interactions with neighboring subsystems. Recently, prior knowledge and data were combined within centralized robust controller design (\cite{Berberich-et-al-2023-Prior-Knowledge-Data-Robust-Design}). Here, we address the intersection of gray-box design and distributed synthesis of a structured controller.

\subsection{Related Work}

The topic of distributed model-based controller synthesis for interconnected systems has been addressed comprehensively in the literature. For example, \cite{De-Pasquale-2020-Extended-Full-Block-S-Procedure-Distributed-Control-Interconnected-Systems} and \cite{Sturz-et-al-2021-Distributed-Control-Design-Heterogeneous-Interconnected} developed approaches for homogeneous and heterogeneous interconnected systems, respectively. \cite{Conte-et-al-2012-Distributed-Synthesis-Control-Constrained-Linear-Sys} proposed a distributed synthesis scheme for the control of linear systems subject to trajectory constraints.

In recent years, a wide variety of direct data-driven control approaches have been proposed, including state feedback control (\cite{De-Persis-2020-Formulas-Data-Driven-Control}), data-driven model predictive control (MPC) (\cite{Berberich-Allgower-Review-Data-Driven-MPC}), and output feedback control (\cite{Li-2024-Controller-Synthesis-Noisy-Input-Noisy-Output-Data}). Most data-driven methods for distributed control solve a system-wide synthesis problem, such as \cite{Miller-et-al-2025-Data-Driven-Structured-Robust-Control-Linear-Sys} and \cite{Yang-et-al-2025-Data-Driven-Structured-Controller-Design-S-Procedure}. On the other hand, \cite{Kohler-2022-Data-Driven-Distributed-MPC-Coupled-Linear} undertake decentralized terminal controller synthesis and assume sufficiently weak coupling between subsystems. A black-box data-driven distributed robust controller synthesis approach is proposed by \cite{Steentjes-2022-PhD-Thesis}, which is unable to utilize any available model knowledge. Meanwhile, data-driven system-level synthesis (SLS) (\cite{Xue-Matni-2021-Data-Driven-SLS}), with an MPC variant (\cite{Alonso-2022-Data-Driven-Distributed-Localized-MPC}), may be able to handle gray-box design but could suffer from a poor trade-off between performance and communication overheads due to the enforcement of localized responses.

\subsection{Main Contributions}

In this paper, we propose a distributed synthesis approach to combine model knowledge and data within \(\mathcal{H}_2\) controller design. The system dynamics are assumed to be linear and full state feedback is assumed available. We extend recent results in data-driven control (\cite{Berberich-et-al-2023-Prior-Knowledge-Data-Robust-Design}) to the distributed setting; to the best of the authors' knowledge, this is the first work to propose distributed gray-box controller design based on this line of theory. Under certain sufficient conditions, we decompose the controller synthesis problem so that it is amenable to distributed optimization, thus enabling scalable distributed synthesis. 
We employ the Alternating Direction Method of Multipliers (ADMM) (\cite{Boyd-2011-ADMM}) to solve the problem in an iterative fashion, and detail a scheme that does not require a central server within the solution workflow.



\subsection{Notation}

We use \(\mathbbm{R}\) and \(\mathbbm{Z_+}\) to denote the sets of real numbers and positive integers, respectively. We denote the vertical concatenation of vectors \(v_i\) for \(i \in \mathcal{I}\) by \([v_i]_{i \in \mathcal{I}}\). The \(n \times n\) identity matrix is denoted by \(I_n\), where the dimension may be omitted if it can be inferred from context.

\section{Problem Formulation}
\label{section:problem_formulation}


Consider a linear time-invariant (LTI) dynamical system made up of \(M \in \mathbbm{Z}_+\) interconnected subsystems, coupled via their states. We denote the set of subsystems by \(\mathcal{V} := \{1,\ldots,M\}\). The structure of the physical interconnections can be represented by a directed graph \(\mathcal{G}_P := (\mathcal{V}, \mathcal{E}_P)\), where \(\mathcal{E}_P \subseteq \mathcal{V} \times \mathcal{V}\) is the graph edge set. An edge \((j,i) \in \mathcal{E}_P\) exists if the state trajectory of subsystem \(i \in \mathcal{V}\) is directly influenced by the state of subsystem \(j \in \mathcal{V}\), with \(i \neq j\). We define the neighborhood of \(i\) as the neighbors of \(i\) in \(\mathcal{G}\) as well as \(i\) itself; \(\mathcal{N}_i := \{j : (j,i) \in \mathcal{E}_P\} \ \cup \ \{i\}\). Each subsystem is controlled by an independent control agent, so we may also denote the set of controlling agents by \(\mathcal{V}\). We describe the LTI dynamics of each subsystem by:

\begin{subequations}
    \begin{align}
    \label{eqn:subsys_dynamics}
    x_i^{k+1} &= A_{\mathcal{N}_i} x_{\mathcal{N}_i}^k + B_i u_i^k + B_{d,i} d_i^k \\
    \label{eqn:subsys_output}
    y_i^k &= C_{y,\mathcal{N}_i} x_{\mathcal{N}_i}^k + D_{y,i} u_i^k,
\end{align}
\end{subequations}

where for subsystem \(i \in \mathcal{V}\) and at time-step \(k\), the subsystem state is given by \(x_i^k \in \mathbbm{R}^{n_i}\), the subsystem input is given by \(u_i^k \in \mathbbm{R}^{m_i}\), \(d_i^k \in \mathbbm{R}^{n_{d,i}}\) is a stochastic process disturbance, \(x_{\mathcal{N}_i}^k \in \mathbbm{R}^{n_{\mathcal{N}_i}}\) is given by \(x_{\mathcal{N}_i}^k = [x_j^k]_{j \in \mathcal{N}_i}\), and \(y_i^k \in \mathbbm{R}^{p_i}\) is a performance output signal. The dynamics of the global system may be written as:

\begin{subequations}
    \begin{align}
    \label{eqn:global_dynamics}
    x^{k+1} &= Ax^k + Bu^k + B_d d^k \\
    \label{eqn:global_output}
    y^k &= \Tilde{C}_y x^k + \Tilde{D}_y u^k,
\end{align}
\end{subequations}

where at time-step \(k\), the system state is given by \(x^k = [x_i^k]_{i \in \mathcal{V}} \in \mathbbm{R}^n\), the global input is given by \(u^k = [u_i^k]_{i \in \mathcal{V}} \in \mathbbm{R}^m\), the global disturbance is \(d^k = [d_i^k]_{i \in \mathcal{V}}\), and \(y^k \in \mathbbm{R}^p\) is a global performance output signal.  Note that \(B\) and \(B_d\) are block-diagonal, whilst \(A\) is structured by \(\mathcal{G}_P\). To simplify notation, we introduce lifting matrices \(T_i \in \{0,1\}^{n_i \times n}\), \(G_i \in \{0,1\}^{m_i \times m}\), and \(W_i \in \{0,1\}^{n_{\mathcal{N}_i} \times n}\) such that \(x_i^k = T_i x^k\), \(u_i^k = G_i u^k\), and \(x_{\mathcal{N}_i}^k = W_i x^k\). \(T_i\), \(G_i\), and \(W_i\) have exactly one entry in each row equal to 1, whilst all other entries are zero. We assume that whilst the system structure is known, the state-space model parameters are partially or fully unknown. 


\begin{assum}
    The graph \(\mathcal{G}_P\) is known and all subsystems \(i \in \mathcal{V}\) follow the dynamics \eqref{eqn:subsys_dynamics}. The global pair \((A,B)\) is stabilizable and \(B_i \neq \mathbf{0} \ \forall i \in \mathcal{V}\). The parameters of the state space matrices \(A_{\mathcal{N}_i}\), \(B_i\), and \(B_{d,i}\) in \eqref{eqn:subsys_dynamics} are partially or fully unknown \( \forall i \in \mathcal{V}\).
\end{assum}

Our goal is to design a state feedback controller \(K_{\mathcal{N}_i}\) for each subsystem, where the subsystem control input is calculated by the control law \(u_i^k = K_{\mathcal{N}_i} x_{\mathcal{N}_i}^k\), such that the global system is closed-loop stable and has a bounded \(\mathcal{H}_2\)-norm. We denote a global controller by \(K = \sum_{i \in \mathcal{V}}G_i^\top K_{\mathcal{N}_i} W_i\). The \(\mathcal{H}_2\) design problem is given in \cite{Scherer-Weiland-2011-LMIs-in-Control} Proposition 3.13 and is as follows: if there exists \(P \succ 0\), \(K\) such that:

\begin{subequations}
\label{eqn:H2_textbook_conditions}
    \begin{align}
    (A + BK)^\top P (A + BK) - P + \overline{C}^\top \overline{C} &\prec 0 \\
    \mathrm{trace}(B_d^\top P B_d) &< \gamma^2,
\end{align}
\end{subequations}

then the \(\mathcal{H}_2\) norm of the mapping \(d \mapsto y\) is bounded by \(\gamma\), where \(y^k = \overline{C} x^k\) in closed-loop. Note that we assume no direct feedthrough from \(d^k\) to \(y^k\) (cf. \eqref{eqn:global_output}), which is a necessary and sufficient condition for the possibility of a finite closed-loop \(\mathcal{H}_2\) norm (\cite{Scherer-Weiland-2011-LMIs-in-Control}). We propose a distributed synthesis solution scheme based on ideas in \cite{Jokic-Lazar-2009-Decentralized-Stabilization-Nonlinear, Conte-2016-Distributed-Synthesis-Stability-Cooperative-Distributed-MPC-Linear}.

We cannot directly tackle the problem of finding a controller and Lyapunov function since the dynamics are not fully known. Partially unknown dynamics may be represented by a Linear Fractional Transformation (LFT) of the form (\cite{Berberich-et-al-2023-Prior-Knowledge-Data-Robust-Design}):

\begin{subequations}
\label{eqn:LFT}
    \begin{align}
        \label{eqn:LFT_dynamics}
        \left[
        \begin{array}{c}
            x_i^{k+1} \\
            \hline y_i^k \\ z_i^k
        \end{array}
        \right] &= \left[\begin{array}{c|c c c}
            A_{\mathcal{N}_i}' & B_i' & B_{d,i} & B_{w,i} \\
            \hline C_{y,\mathcal{N}_i} & D_{y,i} & 0 & 0 \\ C_{z,\mathcal{N}_i} & D_{z,i} & 0 & 0
        \end{array}
        \right] \left[
        \begin{array}{c}
            x_{\mathcal{N}_i}^k \\
            \hline u_i^k \\ d_i^k \\ w_i^k
        \end{array}
        \right] \\
        w_i^k &= \Delta_{i}^\mathrm{tr} z_i^k,
    \end{align}
\end{subequations}

where \(A_{\mathcal{N}_i}'\), \(B_i'\), and \(B_{d,i}\) represent the known parts of the dynamics \eqref{eqn:subsys_dynamics}, \(y_i^k \in \mathbbm{R}^{p_i}\) is a performance output signal, and the variables \(w_i^k \in \mathbbm{R}^{n_{w,i}}\) and \(z_i^k \in \mathbbm{R}^{n_{z,i}}\) represent an uncertainty channel. 

\begin{assum}
    All matrices in \eqref{eqn:LFT} are known, except for the true uncertainty \(\Delta_i^\mathrm{tr}\) capturing the unknown dynamics.
\end{assum}

\begin{rem}
    If \(B_{d,i}\) is not fully known or is not of full column rank, the disturbance may be redefined as \(\Tilde{d}_i^k = B_{d,i} d_i^k\), then we set \(B_{d,i} = I\) in \eqref{eqn:LFT_dynamics}, at the cost of conservatism.
\end{rem} 

In order to gain an understanding of the uncertain system behavior and deal with the unknown dynamics, we take a data-driven approach.
%
%
We assume the availability of a dataset that fulfills the following global condition:

\begin{assum}
\label{assum:PE}
    An input-state dataset \((\{_du^k\}_{k=0:N-1}, \\ \{_dx^k\}_{k=0:N})\) is available and is such that \\ \(\begin{bmatrix}
        _du^0 & _du^1 & \ldots & _du^{N-1} \\ _dx^0 & _dx^1 & \ldots & _dx^{N-1} 
    \end{bmatrix}\) is of full row rank.
\end{assum}

The subscript \(d\) in \(_du^k\) and \(_dx^k\) denotes a dataset sample of \(u^k\) and \(x^k\), respectively. 
We distribute the global dataset among the control agents according to the structure of \(\mathcal{G}_P\), and define the following data matrices for each agent:
\begin{subequations}
\label{eqn:data_matries_def}
    \begin{align}
    &U_i := \begin{bmatrix}
        _du_i^0 & _du_i^1 \ldots & _du_i^{N-1}
    \end{bmatrix} \\ &X_i^+ := \begin{bmatrix}
        _dx_i^1 & _dx_i^2 \ldots & _dx_i^{N}
    \end{bmatrix} \\ & X_{\mathcal{N}_i} := \begin{bmatrix}
        _dx_{\mathcal{N}_i}^0 & _dx_{\mathcal{N}_i}^1 \ldots & _dx_{\mathcal{N}_i}^{N-1}
    \end{bmatrix} \\ 
    \label{eqn:Z_i_def}
    & Z_i := C_{z,\mathcal{N}_i} X_{\mathcal{N}_i} + D_{z,i} U_i \\
    \label{eqn:M_i_def}
    & M_i := X_i^+ - A_{\mathcal{N}_i}' X_{\mathcal{N}_i} - B_i' U_i.
\end{align}
\end{subequations}

We assume that we have prior knowledge of bounds on the disturbance, and that some prior knowledge on the uncertainty is available. We focus on quadratic full-block bounds to describe the disturbance and norm bounds to describe prior knowledge on \(\Delta_i^\mathrm{tr}\). 

\begin{rem}
    A variety of disturbance descriptions may be considered, leading to, for example, diagonal and convex hull multipliers, at the expense of increased computational effort. Structured prior knowledge descriptions such as repeated scalar blocks may also be considered. We refer to \cite{Berberich-et-al-2023-Prior-Knowledge-Data-Robust-Design} for further details.
\end{rem}

During data collection, the disturbance signal generates an unknown sequence \(\{_dd_i^k\}_{k=0:N-1}\), which may be written as a matrix \(D_i^\mathrm{tr} = \begin{bmatrix}
    _dd_i^0 & _dd_i^1 & \ldots & _dd_i^{N-1}
\end{bmatrix}\).

\begin{assum}
\label{assum:disturbance_bound}
    The unknown matrix \(D_i^\mathrm{tr}\) satisfies \(D_i^\mathrm{tr} \in \mathbf{D}_i\), where \(\mathbf{D}_i := \bigg\{ D_i \, \bigg| \, \begin{bmatrix}
                D_i^\top \\ I
            \end{bmatrix}^\top P_{d,i} \begin{bmatrix}
                D_i^\top \\ I
            \end{bmatrix} \succeq 0 \ \forall P_{d,i} \in \mathbf{P}_{d,i} \bigg\},\)
%
%
    with \(\mathbf{P}_{d,i}\) defined by:

    \begin{equation}
    \label{eqn:P_d_i_def_disturbance_full_block_quad}
        \mathbf{P}_{d,i} := \left\{ \left. \tau_i \begin{bmatrix}
            Q_{d,i} & S_{d,i} \\ S_{d,i}^\top & R_{d,i} 
        \end{bmatrix} \, \right| \, \tau_i > 0 \right\},
    \end{equation}

    where \(Q_d \prec 0\).
\end{assum}

Our prior knowledge on the uncertainty is assumed to take the following form for each subsystem:

\begin{assum}
\label{assum:prior_knowledge_set}
    The true uncertainty lies within a prior uncertainty set \(\mathbf{\Delta}_{\mathrm{pr},i}\) such that \(\Delta_i^\mathrm{tr} \in \mathbf{\Delta}_{\mathrm{pr},i} \ \forall i \in \mathcal{V}\), where \(        \mathbf{\Delta}_{\mathrm{pr},i} := \{\Delta_i : \Delta_i \Delta_i^\top \preceq \overline{\delta}_i I\},\)
 %
%
    for some known \(\overline{\delta}_i > 0\).
\end{assum}

These prior knowledge bounds on the disturbance and the uncertainty together with the data will be used in the sequel to bound the overall uncertainty set. Due to the physical coupling between subsystems, we require communication between agents to ensure that we can find a globally stabilizing controller. We assume that bidirectional communication exists between any pair of agents controlling subsystems that are connected by an edge in \(\mathcal{G}_P\). We denote an undirected cyber-layer graph by \(\mathcal{G}_C := (\mathcal{V}, \mathcal{E}_C)\), where the edge set \(\mathcal{E}_C\) is defined by \(\mathcal{E}_C := \{\{i,j\} : (i,j) \in \mathcal{E}_P \ \lor \ (j,i) \in \mathcal{E}_P\}\).

\begin{assum}
    Any pair of control agents \(i,j \in \mathcal{V}\) have bidirectional communication if \(\{i,j\} \in \mathcal{E}_C\). Network phenomena, e.g., packet loss, are negligible and latency is negligible compared to the timescale of the system dynamics.
\end{assum}

We are now equipped to design a distributed \(\mathcal{H}_2\) state-feedback controller for partially or fully unknown LTI systems through distributed synthesis, utilizing available model knowledge, data, and prior knowledge on the disturbance and unknown dynamics.

\section{Main Results}
\label{section:main_results}

\subsection{Decomposing \(\mathcal{H}_2\) Design Conditions}

We begin by finding a decomposition of the global \(\mathcal{H}_2\) design conditions \eqref{eqn:H2_textbook_conditions} that can be made amenable to distributed synthesis. Conditions that a decomposable Lyapunov function must satisfy to guarantee closed-loop stability are detailed in \cite{Jokic-Lazar-2009-Decentralized-Stabilization-Nonlinear}, which reduce to matrix inequalities for the linear dynamics \eqref{eqn:subsys_dynamics} and assuming quadratic functions (\cite{Conte-2016-Distributed-Synthesis-Stability-Cooperative-Distributed-MPC-Linear}).

\begin{lem}
    If there exist, \(\forall i \in \mathcal{V}\), a \(\gamma_i > 0\), \(P_i \succ 0\), \(K_{\mathcal{N}_i}\), \(\Theta_{\mathcal{N}_i}\) such that the following inequalities hold, then the control law \(u^k = K x^k\), where \(K = \sum_{i=1:M} G_i^\top K_{\mathcal{N}_i} W_i\), stabilizes the dynamics \eqref{eqn:global_dynamics}, and when \(u^k = K x^k\), the mapping \(d \mapsto y\) has an \(\mathcal{H}_2\)-norm bounded by \(\gamma\), where \(\gamma^2 = \sum_{i=1:M} \gamma_i^2\) and \(y^k = \sum_{i=1:M} (C_{y,\mathcal{N}_i} + D_{y,i} K_{\mathcal{N}_i}) W_i x^k\):

    \begin{subequations}
    \label{eqn:distributed_H2_global_conditions}
        \begin{align}
            \begin{aligned}
            \label{eqn:Lyap_inequality_neighbourhood}
                (A_{\mathcal{N}_i} + B_i K_{\mathcal{N}_i})^\top P_i (A_{\mathcal{N}_i} + B_i K_{\mathcal{N}_i}) - \overline{P}_i & \\ + (C_{y,\mathcal{N}_i} + D_{y,i} K_{\mathcal{N}_i})^\top (C_{y,\mathcal{N}_i} + D_{y,i} K_{\mathcal{N}_i} &) \prec \Theta_{\mathcal{N}_i}, 
            \end{aligned} \\
            \label{eqn:H2_performance_bound_true}
            \mathrm{trace}(B_{d,i}^\top P_i B_{d,i})  < \gamma_i^2, \ \ & \\
            \label{eqn:relaxation_inequality_global}
            \sum_{i=1}^M W_i^\top \Theta_{\mathcal{N}_i} W_i \preceq 0, \ \ &
        \end{align}
    \end{subequations}

    where \(\overline{P}_i := W_i T_i^\top P_i T_i W_i^\top\).
\end{lem}

\begin{pf}
    Let \(V(x^k) = \sum_{i=1}^M V_i(x_i^k)\), where \(V_i(x_i^k) = x_i^{k^\top} P_i x_i^k\), let \(l_i(x_{\mathcal{N}_i}^k, u_i^k) = 
    x_{\mathcal{N}_i}^{k^\top} \Tilde{C}_{\mathcal{N}_i}^\top \Tilde{C}_{\mathcal{N}_i} x_{\mathcal{N}_i}^k\), where \(\Tilde{C} = (C_{y,\mathcal{N}_i} + D_{y,i} K_{\mathcal{N}_i})\) and \(u_i^k\) is calculated according to the control law \(u_i^k = K_{\mathcal{N}_i} x_{\mathcal{N}_i}^k\) (cf. \eqref{eqn:subsys_output}), and let \(\theta_i(x_{\mathcal{N}_i}^k) = x_{\mathcal{N}_i}^{k^\top} \Theta_{\mathcal{N}_i} x_{\mathcal{N}_i}^k\). Then, if the following inequalities hold, \(V(x^k)\) is a valid Lyapunov function such that \(V(x^{k+1}) - V(x^k) \leq - \sum_{i=1:M} l_i(x_{\mathcal{N}_i}^k, K_{\mathcal{N}_i} x_{\mathcal{N}_i}^k)\) (\cite{Jokic-Lazar-2009-Decentralized-Stabilization-Nonlinear}):

    \begin{subequations}
    \label{eqn:Jokic_inequality_Lyap}
        \begin{align}
            \! \! \! \! V_i(x_i^{k+1}) - V_i(x_i^k) &\leq -l_i(x_{\mathcal{N}_i}^k, K_{\mathcal{N}_i} x_{\mathcal{N}_i}^k) + \theta_{\mathcal{N}_i}(x_{\mathcal{N}_i}^k), \\
            \sum_{i=1:M} \theta_i(x_{\mathcal{N}_i}^k) &\leq 0.
    \end{align}
    \end{subequations}

    Under nominal dynamics in the standard \(\mathcal{H}_2\) setting and using the definition of a negative semidefinite matrix, \eqref{eqn:Jokic_inequality_Lyap} is equivalent to \eqref{eqn:Lyap_inequality_neighbourhood} and \eqref{eqn:relaxation_inequality_global}. Let \(P = \mathrm{diag}(P_i)_{i=1:M}\). We can express \(V(x^k) = x^{k^\top} P x^k\), and under nominal global dynamics, \(V(x^{k+1}) = x^{k^\top} (A + BK)^\top P (A+BK) x^k\). We thus infer that if \eqref{eqn:Lyap_inequality_neighbourhood} and \eqref{eqn:relaxation_inequality_global} hold, then \((A + BK)^\top P (A + BK) - P + \overline{C}^\top \overline{C} \prec 0\),
%
%
    where we let \(\overline{C}^\top \overline{C} = \sum_{i=1:M} W_i^\top \Tilde{C}_{\mathcal{N}_i}^\top \Tilde{C}_{\mathcal{N}_i} W_i\). Since \(B_d\) and \(P\) are block-diagonal, \(\mathrm{trace}(B_d^\top P B_d) = \\ \sum_{i=1:M} \mathrm{trace} (B_{d,i}^\top P_i B_{d,i})\), hence if \(\mathrm{trace} (B_{d,i}^\top P_i B_{d,i}) < \gamma_i^2 \ \forall i\), \(\mathrm{trace}(B_d^\top P B_d) < \gamma^2\), where \(\gamma^2 = \sum_{i=1:M} \gamma_i^2\). Therefore a feasible solution to \eqref{eqn:distributed_H2_global_conditions} satisfies \eqref{eqn:H2_textbook_conditions}, which by \cite{Scherer-Weiland-2011-LMIs-in-Control} Proposition 3.13 guarantees a stable closed-loop response \(d \mapsto y\) with an \(\mathcal{H}_2\)-norm bounded by \(\gamma\) under feedback control \(u^k = Kx^k\), thus completing the proof.
    \qed
\end{pf}

The unknown dynamics remain, preventing us from directly tackling the control synthesis problem. We turn to data-driven robust control techniques to deal with these.

\subsection{Distributed Multipliers from Prior Knowledge and Data}

We extend multiplier classes proposed for the centralized setting by \cite{Berberich-et-al-2023-Prior-Knowledge-Data-Robust-Design} to the distributed setting. Considering \eqref{eqn:LFT} and the definitions of the data matrices \eqref{eqn:data_matries_def}, we see that \(X_i^+ = A_{\mathcal{N}_i}' X_{\mathcal{N}_i} + B_i' U_i + B_{d,i} D_i^\mathrm{tr} + B_{w,i} \Delta_i^\mathrm{tr} (C_{z,\mathcal{N}_i} X + D_{z,i} U_i)\). Considering the definitions of \(M_i\) \eqref{eqn:M_i_def} and \(Z_i\) \eqref{eqn:Z_i_def}, we define a set of learned uncertainties by \(\mathbf{\Delta}_i^\mathrm{learned} := \{ \Delta_i \ | \ M_i = B_{w,i} \Delta_i Z_i + B_{d,i} D_i \ \mathrm{s.t.} \ D_i \in \mathbf{D}_i \}.\)
%
%
Evidently, \(\Delta_i^\mathrm{tr} \in \mathbf{\Delta}_i^\mathrm{learned}\). We define a distributed multiplier class by:

\begin{equation}
    \mathbf{\Tilde{P}}_{d,i} := \begin{bmatrix}
        -Z_i & 0 \\ M_i & B_{d,i}
    \end{bmatrix} \mathbf{P}_{d,i} \begin{bmatrix}
        -Z_i & 0 \\ M_i & B_{d,i}
    \end{bmatrix}^\top,
\end{equation}

where \(\mathbf{P}_{d,i}\) is given by \eqref{eqn:P_d_i_def_disturbance_full_block_quad}. 


We now define a multiplier class for our prior knowledge given in Assumption \ref{assum:prior_knowledge_set}. The inequality \(\Delta_i \Delta_i^\top \preceq \overline{\delta}_i I\) may be written as a quadratic form, which leads to the multiplier class for the transformed uncertainty \(\Tilde{\Delta}_i := B_{w,i} \Delta_i\):

\begin{equation*}
    \mathbf{\Tilde{P}}_{\mathrm{pr},i} := \left\{ P_{\mathrm{pr},i} \, \left| \, \lambda_i \begin{bmatrix}
        I & 0 \\ 0 & B_{w,i}
    \end{bmatrix} \begin{bmatrix}
        -I & 0 \\ 0 & \overline{\delta}_i
    \end{bmatrix} \begin{bmatrix}
        I & 0 \\ 0 & B_{w,i}^\top
    \end{bmatrix}, \, \lambda_i \geq 0 \right. \right\}.
\end{equation*}

%
%

%
%

Under both Assumptions \ref{assum:disturbance_bound} and \ref{assum:prior_knowledge_set}, by \cite{Berberich-et-al-2023-Prior-Knowledge-Data-Robust-Design} Lemma 3, the data-driven and prior knowledge-based multipliers can be combined exactly. 
A transformed combined multiplier set is defined by \(\mathbf{\Tilde{P}}_{\mathrm{com},i} := \mathbf{\Tilde{P}}_{d,i} \oplus \mathbf{\Tilde{P}}_{\mathrm{pr},i}\), with the uncertainty set \(\mathbf{\Tilde{\Delta}}_{i}^\mathrm{com}\) given by:

\begin{equation*}
    \mathbf{\Tilde{\Delta}}_{i}^\mathrm{com} = \left\{ \Tilde{\Delta}_i \left| \begin{bmatrix}
        \Tilde{\Delta}_i^\top \\ I
    \end{bmatrix}^\top \! \! \Tilde{P}_{\mathrm{com},i} \begin{bmatrix}
        \Tilde{\Delta}_{i}^\top \\ I
    \end{bmatrix} \succeq 0 \, \forall \Tilde{P}_{\mathrm{com},i} \in \mathbf{\Tilde{P}}_{\mathrm{com},i} \right. \right\}.
\end{equation*}

In the following analysis, we will work with uncertainty blocks of the form \(\begin{bmatrix}
    I & \Delta_i
\end{bmatrix}\) and its transpose. Therefore, we define a mirrored transformed combined multiplier:

\begin{equation}
    \mathbf{\overline{P}}_{\mathrm{com},i} := \begin{bmatrix}
            0 & I \\ I & 0
        \end{bmatrix}^\top \mathbf{\Tilde{P}}_{\mathrm{com},i} \begin{bmatrix}
            0 & I \\ I & 0
        \end{bmatrix}.
\end{equation}


\subsection{Distributed Controller Synthesis}

Having combined our prior knowledge and information learned from data in a compact manner, we are now ready to formulate a Linear Matrix Inequality (LMI) design condition to find a distributed \(\mathcal{H}_2\) controller in a fashion that is amenable to distributed optimization.

\begin{thm}
\label{thm:distributed_synthesis_grey_box}
    Suppose that for all \(i \in \mathcal{V}\), there exist \(E_i \succ 0, E_{\mathcal{N}_i} \succ 0, \Phi_i \succ 0\), \(\overline{P}_{\mathrm{com},i} \in \mathbf{\overline{P}}_{\mathrm{com},i}\), \(\gamma_i > 0\), \(S_{\mathcal{N}_i}\), \(F_{\mathcal{N}_i}\), \(Y_{\mathcal{N}_i}\) such that:

    {\allowdisplaybreaks
    \begin{subequations}
    \label{eqn:distributed_synthesis_LMIs}
        \begin{align}
        \label{eqn:H2_performance_bound}
            \mathrm{trace}(\Phi_i) &< \gamma_i^2 \\
            \label{eqn:H2_performance_LMI}
            \begin{bmatrix}
                \Phi_i & B_{d,i}^\top \\
                * & E_{_i}
            \end{bmatrix} &\succ 0
        \end{align}
        \begin{align}
            \label{eqn:distributed_stability_LMI}
            \begin{bmatrix}
                \overline{P}_{\mathrm{com},i} + \begin{bmatrix}
                    - E_i & 0 \\ 0 & 0
                \end{bmatrix} & * & * \vspace{1mm} \\
                \begin{bmatrix}
                    A_{\mathcal{N}_i}' E_{\mathcal{N}_i} + B_i' Y_{\mathcal{N}_i} \\
                    C_{z,\mathcal{N}_i} E_{\mathcal{N}_i} + D_{z,i} Y_{\mathcal{N},i}
                \end{bmatrix}^\top & -\overline{E} - F_{\mathcal{N}_i} & * \\
                0 & C_{y,\mathcal{N}_i} E_{\mathcal{N}_i} + D_{y,i} Y_{\mathcal{N}_i} & -I
            \end{bmatrix} \prec 0
        \end{align}
        \begin{align}
            \label{eqn:F_less_than_S}
            F_{\mathcal{N}_i} &\preceq S_{\mathcal{N}_i} \\
            \label{eqn:neighbourhood_S_coupling}
            \sum_{j \in \mathcal{N}_i} T_i W_j^\top S_{\mathcal{N}_j} W_j T_i^\top &\preceq 0
        \end{align}
    \end{subequations}
    }
    
    where \(E_{\mathcal{N}_i} := \mathrm{diag}(E_j)_{j \in \mathcal{N}_i}\), \(\overline{E}_i := W_i T_i^\top E_i T_i W_i^\top\), and \(S_{\mathcal{N}_i}\) is block-diagonal for all \(i \in \mathcal{V}\). Then the conditions \eqref{eqn:distributed_H2_global_conditions} hold with \(P_i = E_i^{-1}\), \(K_{\mathcal{N}_i} = Y_{\mathcal{N}_i} E_{\mathcal{N}_i}^{-1}\), \(\Theta_{\mathcal{N}_i} = E_{\mathcal{N}_i}^{-1} F_{\mathcal{N}_i} E_{\mathcal{N}_i}^{-1}\).
\end{thm}

\begin{pf}
     Applying the Schur complement to the lower-right block of \eqref{eqn:distributed_stability_LMI} and pre- and post-multiplying \eqref{eqn:distributed_stability_LMI} by \(\mathrm{diag}(I,E_{\mathcal{N}_i}^{-1})\) results in

    \begin{equation}
    \label{eqn:stability_LMI_E_Ni_multiply}
        \begin{bmatrix}
            \overline{P}_{\mathrm{com},i} + \begin{bmatrix}
                    - E_i & 0 \\ 0 & 0
                \end{bmatrix} & * \\
                \begin{bmatrix}
                    A_{\mathcal{N}_i}' + B_i' K_{\mathcal{N}_i} \\
                    C_{z,\mathcal{N}_i} + D_{z,i} K_{\mathcal{N}_i}
                \end{bmatrix}^\top & -\mathcal{M}_{\mathcal{N}_i}
        \end{bmatrix} \prec 0,
    \end{equation}

    where \(\mathcal{M}_{\mathcal{N}_i} = \Theta_{\mathcal{N}_i} + \overline{P}_i - (C_{y,\mathcal{N}_i} + D_{y,i} K_{\mathcal{N}_i})^\top (C_{y,\mathcal{N}_i} + D_{y,i} K_{\mathcal{N}_i})\) and \(\overline{P}_i = W_i T_i^\top P_i T_i W_i^\top\). Note that \(E_{\mathcal{N}_i} = W_i E W_i^\top\), where \(E = \mathrm{diag}(E_i)_{i=1:M}\). Applying the Schur complement to the lower-right block of \eqref{eqn:stability_LMI_E_Ni_multiply} and re-arranging terms leads to

    \begin{equation}
    \label{eqn:stability_LMI_Schur}
        \begin{aligned}
            \overline{P}_{\mathrm{com},i} + \begin{bmatrix}
            I & A_{\mathcal{N}_i}' + B_i' K_{\mathcal{N}_i} \\ 0 & C_{z,\mathcal{N}_i} + D_{z,i} K_{\mathcal{N}_i}
        \end{bmatrix} \qquad \qquad \qquad & \\ \cdot \begin{bmatrix}
            - E_i & 0 \\ 0 & \mathcal{M}_{\mathcal{N}_i}^{-1}
        \end{bmatrix} \begin{bmatrix}
            *
        \end{bmatrix} &\prec 0.
        \end{aligned}
    \end{equation}

    By the full-block S-procedure (\cite{Scherer-2001-LPV-Control-Full-Block-Multipliers}), and for dynamics such that \(A_{\mathcal{N}_i} + B_i K_{\mathcal{N}_i} = A_{\mathcal{N}_i}' + B_i' K_{\mathcal{N}_i} + \Delta_i (C_{z,\mathcal{N}_i} + D_{z,i} K_{\mathcal{N}_i})\), \eqref{eqn:stability_LMI_Schur} implies that for all \(\Delta_i\) such that \(B_{w,i} \Delta_i \in \mathbf{\Tilde{\Delta}}_{\mathrm{com},i}\): \((A_{\mathcal{N}_i} + B_i K_{\mathcal{N}_i}) \mathcal{M}_{\mathcal{N}_i}^{-1} (*) - E_i \prec 0.\)
%
%
    Applying the Schur complement twice returns \eqref{eqn:Lyap_inequality_neighbourhood}.
%
%
    Since \(\mathcal{S}_{\mathcal{N}_i}\) is block-diagonal, \eqref{eqn:F_less_than_S} and \eqref{eqn:neighbourhood_S_coupling} are sufficient conditions for \(\sum_{i=1:M} W_i^\top F_{\mathcal{N}_i} W_i \preceq 0\) (\cite{Conte-2016-Distributed-Synthesis-Stability-Cooperative-Distributed-MPC-Linear}). This constraint is equivalent to \eqref{eqn:relaxation_inequality_global}--see the proof of \cite{Conte-et-al-2012-Distributed-Synthesis-Control-Constrained-Linear-Sys} Theorem IV.3 for details. Applying the Schur complement to \eqref{eqn:H2_performance_LMI} and considering \eqref{eqn:H2_performance_bound} gives that \(\mathrm{trace}(B_{d,i}^\top P_i B_{d,i}^\top) < \gamma_i^2\),
%
%
    which is \eqref{eqn:H2_performance_bound_true}. We have recovered the conditions \eqref{eqn:distributed_H2_global_conditions}, thus completing the proof.
    \qed
\end{pf}

The set of LMIs \eqref{eqn:distributed_synthesis_LMIs} depend only on local and neighborhood variables. We have eliminated the globally coupled constraint \eqref{eqn:relaxation_inequality_global}, and have formulated a convex design problem. The controllers can now be synthesized in a distributed fashion using distributed optimization techniques.

\subsection{Distributed Optimization Formulation}
\label{subsection:distributed_optimization}

Each agent \(i\)'s controller synthesis problem \eqref{eqn:distributed_synthesis_LMIs} depends on local variables and their neighbor \(j\)'s variables; \(E_j\) via \(E_{\mathcal{N}_i}\) and \(S_{\mathcal{N}_j}\) via \eqref{eqn:neighbourhood_S_coupling}. In order to optimize global performance, we aim to minimize \(\gamma\). Since \(\gamma^2 = \sum_{i=1:M} \gamma_i^2\), we have a separable objective. The goal of minimizing \(\gamma\) is equivalent to the goal of minimizing \(\sum_{i=1:M} \mathrm{trace}(\Phi_i)\). The global problem can be written as:

\begin{subequations}
\label{eqn:global_optimization_problem}
    \begin{align}
        \min_{\{\Phi_i,E_i \succ 0, \overline{P}_{\mathrm{com},i} \in \mathbf{\overline{P}}_{\mathrm{com},i},S_{\mathcal{N}_i}, F_{\mathcal{N}_i},Y_{\mathcal{N}_i}\}_{i \in \mathcal{V}}} & \sum_{i \in \mathcal{V}} \mathrm{trace}(\Phi_i) \\
        \mathrm{s.t.} \ \eqref{eqn:H2_performance_LMI} - \eqref{eqn:neighbourhood_S_coupling} & \ \forall i \in \mathcal{V}
    \end{align}
\end{subequations}

By introducing local copies of the neighbor variables, denoted \(E_j^{(i)}\) and \(S_{\mathcal{N}_j}^{(i)}\) for agent \(i\)'s copies of agent \(j\)'s variables, we can formulate a distributed optimization problem, which we solve using ADMM (\cite{Boyd-2011-ADMM}) to reach consensus between neighboring agents' variables. Note that \(E_i^{(i)}\) and \(S_{\mathcal{N}_i}^{(i)}\) denote agent \(i\)'s local variables. We use an edge-based ADMM formulation that avoids the need for a central server, enabling a fully distributed implementation -- the method is given in \cite[Chapter~11.2]{Ryu_Yin_2022-LSCO}, which we refer to for details. To ease notation, we define \(\Xi_i := \{\Phi_i, \{E_j^{(i)}\}_{j \in \mathcal{N}_i}, \overline{P}_{\mathrm{com},i}, \{S_{\mathcal{N}_j}^{(i)}\}_{j \in \mathcal{N}_i}, F_{\mathcal{N}_i}, Y_{\mathcal{N}_i}\}\) as the set of agent \(i\)'s decision variables. We define \(\mathcal{C}_i := \{\Xi_i : \eqref{eqn:H2_performance_LMI} - \eqref{eqn:neighbourhood_S_coupling} \ \mathrm{are \ satisfied}\}\) as a convex set of variables satisfying the LMI constraints, where we replace \(S_{\mathcal{N}_j}\) and \(E_j\) with \(S_{\mathcal{N}_j}^{(i)}\) and \(E_j^{(i)}\), respectively, in \eqref{eqn:H2_performance_LMI}--\eqref{eqn:neighbourhood_S_coupling}. We define \(\iota_{\mathcal{C}_i}(\Xi_i)\) as the indicator function of \(\mathcal{C}_i\), such that \(\iota_{\mathcal{C}_i}(\Xi_i) = 0\) if \(\Xi_i \in \mathcal{C}_i\), \(\iota_{\mathcal{C}_i}(\Xi_i) = \infty\) otherwise. Introducing edge variables \(\nu_{\{i,j\}} \ \forall \{i,j\} \in \mathcal{E}_C\), we can write the global optimization problem \eqref{eqn:global_optimization_problem} as a distributed consensus problem:

\begin{subequations}
    \begin{align}
        &\min_{\{\Xi_i\}_{i \in \mathcal{V}}, \{\nu_{\{i,j\}}\}_{\{i,j\} \in \mathcal{E}_C}} \sum_{i \in \mathcal{V}} (\mathrm{trace}(\Phi_i) + \iota_{\mathcal{C}_i}(\Xi_i)) \\
        \label{eqn:consensus_edge_constraints}
        & \mathrm{s.t.} \  \begin{cases}
            h_{ij}(\Xi_i) - \nu_{\{i,j\}} = 0 \\
            h_{ij}(\Xi_j) - \nu_{\{i,j\}} = 0
        \end{cases} \ \ \forall \{i,j\} \in \mathcal{E}_C,
    \end{align}
\end{subequations}

where we define: 

\begin{equation*}
    h_{ij}(\Xi_i) := \begin{bmatrix}
        \mathrm{vect} \left( E_i^{(i)} \right) \\
        \mathrm{vect} \left( E_j^{(i)} \right) \\
        \mathrm{vect} \left( S_{\mathcal{N}_i}^{(i)} \right) \\
        \mathrm{vect} \left( S_{\mathcal{N}_j}^{(i)} \right)
    \end{bmatrix}, \ h_{ij}(\Xi_j) := \begin{bmatrix}
        \mathrm{vect} \left( E_i^{(j)} \right) \\
        \mathrm{vect} \left( E_j^{(j)} \right) \\
        \mathrm{vect} \left( S_{\mathcal{N}_i}^{(j)} \right) \\
        \mathrm{vect} \left( S_{\mathcal{N}_j}^{(j)} \right)
    \end{bmatrix}.
\end{equation*}

By forming an augmented Lagrangian, we arrive at the following ADMM iterations (\cite{Ryu_Yin_2022-LSCO}):

\allowdisplaybreaks{
\begin{subequations}
\label{eqn:ADMM}
    \begin{align}
    \label{eqn:ADMM_x_update}
        &\begin{aligned}
            \Xi_i^{t+1} = &\argmin_{\Xi_i} \, \mathrm{trace}(\Phi_i) + \iota_{\mathcal{C}_i}(\Xi_i) \\ & + \sum_{j \in \mathcal{N}_i \backslash i} \bigg( \mu_{\{i,j\},i}^{t^\top} \left(h_{ij}(\Xi_i) - \nu_{\{i,j\}}^t \right) \\ & + \frac{\alpha}{2} \left\lVert h_{ij}(\Xi_i) - \nu_{\{i,j\}}^t \right\lVert_2^2 \bigg) \ \forall i \in \mathcal{V}
        \end{aligned} \\
        \label{eqn:ADMM_y_update}
        &\begin{aligned}
            \nu_{\{i,j\}}^{t+1} = & \argmin_{\nu_{\{i,j\}}} \sum_{q = i,j} \bigg( \mu_{\{i,j\},q}^{t^\top} \left( h_{ij}(\Xi_q^{t+1}) - \nu_{\{i,j\}} \right) \\ & + \frac{\alpha}{2} \left\lVert h_{ij}(\Xi_q^{t+1}) - y_{\{i,j\}} \right\rVert \ \forall \{i,j\} \in \mathcal{E}_C
        \end{aligned}  \\
        \label{eqn:ADMM_u_update}
        &\begin{aligned}
            \mu_{\{i,j\},q}^{t+1} = \mu_{\{i,j\},q}^t + \alpha (& h_{ij}(\Xi_i^{t+1}) - \nu_{\{i,j\}}^{t+1} ) \\ & \forall \{i,j\} \in \mathcal{E}_C, \, q = i,j,
        \end{aligned}
    \end{align}
\end{subequations}
}

where \(\mu_{\{i,j\},i}\) is agent \(i\)'s dual variable associated with the edge \(\{i,j\}\), \(\alpha\) is the ADMM penalty parameter, and \(t\) is the ADMM iteration number. First, the agents solve \eqref{eqn:ADMM_x_update} in parallel. Secondly, the update \eqref{eqn:ADMM_y_update}--\eqref{eqn:ADMM_u_update} is implemented in parallel.

\begin{thm}
\label{thm:distributed_opt}
    The solution of the distributed optimization iterations \eqref{eqn:ADMM} asymptotically converges to the solution of the global problem \eqref{eqn:global_optimization_problem}. A convergent solution will satisfy \eqref{eqn:distributed_synthesis_LMIs} for all \(i \in \mathcal{V}\) in a consistent manner such that \(E_i = E_i^{(j)}\), \(S_{\mathcal{N}_i} = S_{\mathcal{N}_i}^{(j)}\) \(\forall j \in \mathcal{N}_i\).
\end{thm}

\begin{pf}
    The asymptotic convergence of \eqref{eqn:ADMM} to the solution of \eqref{eqn:global_optimization_problem} is guaranteed by the properties of ADMM (\cite{Ryu_Yin_2022-LSCO}). A convergent solution is guaranteed to satisfy the constraints \eqref{eqn:consensus_edge_constraints}, such that \(E_i^{(i)} = E_i^{(j)}\), \(E_j^{(i)} = E_j^{(j)}\), \(S_{\mathcal{N}_i}^{(i)} = S_{\mathcal{N}_i}^{(j)}\), and \(S_{\mathcal{N}_j}^{(i)} = E_{\mathcal{N}_j}^{(j)}\) if \(\{i,j\} \in \mathcal{E}_C\).
    \qed
\end{pf}

\section{Numerical Experiments}
\label{section:simulations}

Consider the discretized and linearized swing equations as given in the example of \cite{Alonso-2022-Data-Driven-Distributed-Localized-MPC}, which are used to model small-signal frequency dynamics in electric power transmission systems. We take \(M = 39\) and set the interconnection structure \(\mathcal{G}_P\) according to the IEEE 39-bus power system benchmark case, which is shown in Figure \ref{fig:39-bus-system}.
\begin{figure}[ht]
    \centering
    \includegraphics[width=0.75\linewidth]{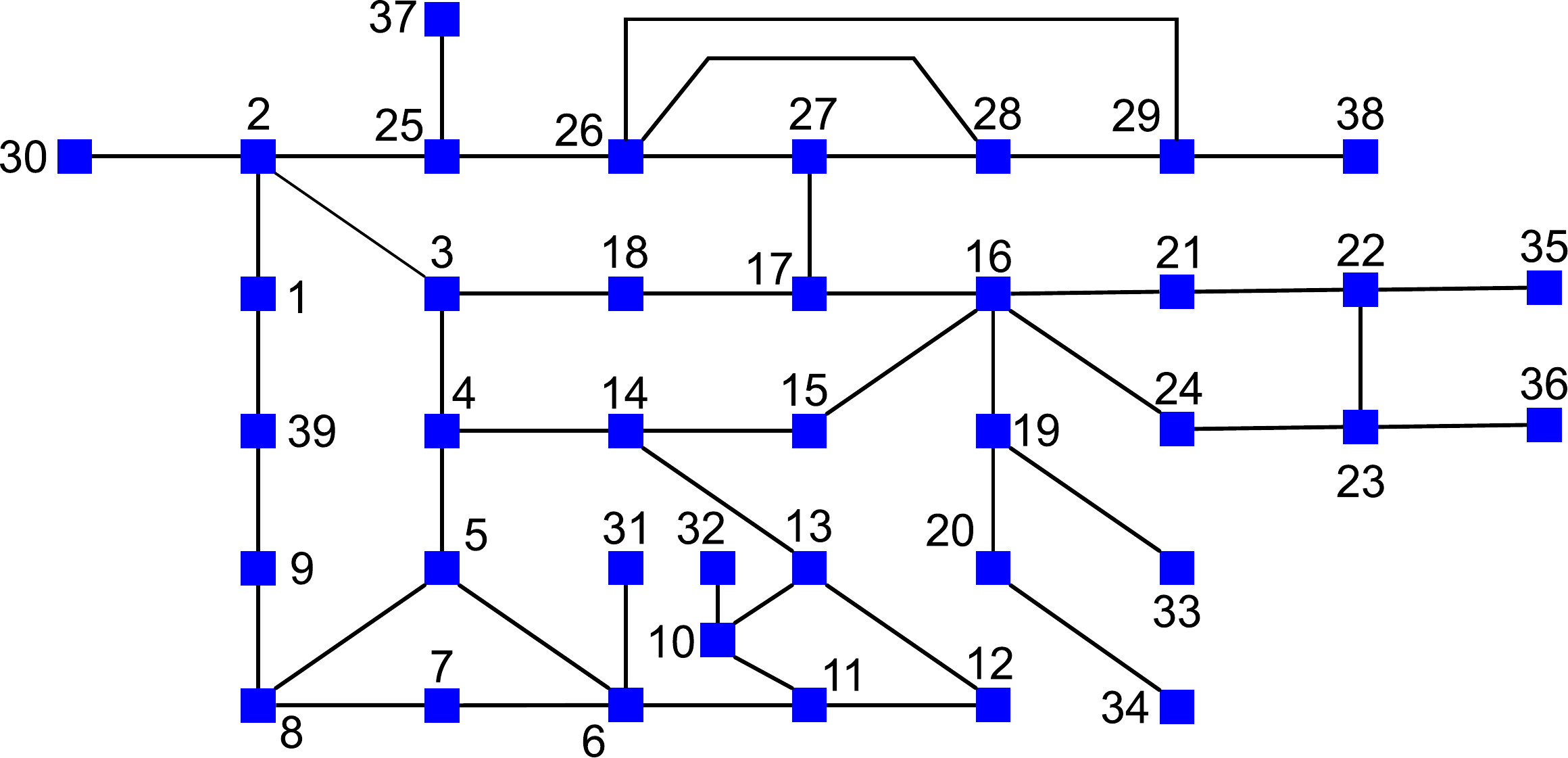}
    \caption{The IEEE 39-bus test system topology}
    \label{fig:39-bus-system}
\end{figure}
The neighborhood dynamics \eqref{eqn:subsys_dynamics} for subsystem \(i\) are written as:

\allowdisplaybreaks{
    \begin{align*}
        \begin{bmatrix}
            \theta_i^{k+1} \\ \omega_i^{k+1}
        \end{bmatrix} &= A_{ii} \begin{bmatrix}
            \theta_i^{k} \\ \omega_i^{k}
        \end{bmatrix} + \sum_{j \in \mathcal{N}_i \backslash i} A_{ij} \begin{bmatrix}
            \theta_j^k \\ \omega_j^k
        \end{bmatrix} + B_i u_i^k + B_{d,i} d_i^k \\
        y_i^k &= \begin{bmatrix}
            Q_r \\ 0
        \end{bmatrix} x_i^k + \begin{bmatrix}
            0 \\ R_r
        \end{bmatrix} u_i^k,
    \end{align*}
}

where {\renewcommand{\arraystretch}{1} \(A_{ii} = \begin{bmatrix}
    1 & \Delta t \\ -\frac{k_i}{m_i} \Delta t & 1 - \frac{c_i}{m_i} \Delta t
\end{bmatrix}\), \(A_{ij} = \begin{bmatrix}
    0 & 0 \\ \frac{k_{ij}}{m_i} \Delta t & 0
\end{bmatrix} \\ \forall j \in \mathcal{N}_i \backslash i\), \(B_i = B_{d,i} = \begin{bmatrix}
    0 & \frac{\Delta t}{m_i}
\end{bmatrix}^\top\)}, \(\Delta t = \qty{0.2}{\second}\) is the discretization interval, \(k_i = \sum_{j \in \mathcal{N}_i \backslash i} k_{ij}\), and we have chosen a linear quadratic performance criterion. The values of \(m_i,c_i,k_{ij}\) are chosen randomly and uniformly in the ranges \([0,2]\), \([0.5,1]\), and \([1,1.5]\), respectively. We simulate a disturbance \(d_i^k \sim \mathcal{U}[-\epsilon_i,\epsilon_i]\), which is compatible with a norm bound on the unknown noise sequence in the dataset \(D_i^\mathrm{tr}\), such that \(\sum_{k=0}^{N-1} \lVert _dd_i^k \rVert_2^2 \leq \overline{d}_i\), where \(\overline{d}_i = N n_{d,i} \epsilon_i^2\) is a worst-case bound. We accordingly form the matrix \(P_{d,i}\) in Assumption \ref{assum:disturbance_bound} with \(Q_{d,i} = -I\), \(R_{d,i} = \overline{d}_i I\), and \(S_{d,i} = 0\) (\cite{Berberich-et-al-2023-Prior-Knowledge-Data-Robust-Design}). We choose \(\overline{\delta}_i \ \forall i\) to be 20\% higher than the true \(\lVert \Delta_i \rVert_2^2\). The dataset is collected by sampling \(_du^k_i \sim \mathcal{U}[-1,1]\) with \(N = 400\). We compare five methods of controller design:

\begin{enumerate}[(i)]
    \item  Distributed black-box synthesis of a distributed controller using Theorems \ref{thm:distributed_synthesis_grey_box} and \ref{thm:distributed_opt}.
    \item Distributed gray-box synthesis of a distributed controller using Theorems \ref{thm:distributed_synthesis_grey_box} and \ref{thm:distributed_opt}, where we assume local dynamics \(A_{ii}\), \(B_i\) and \(B_{d,i}\) are known.
    \item Centralized black-box design of a centralized controller using \cite{Berberich-et-al-2023-Prior-Knowledge-Data-Robust-Design}.
    \item Centralized gray-box design of a centralized controller using \cite{Berberich-et-al-2023-Prior-Knowledge-Data-Robust-Design} with a block-diagonal \(\Delta\).
    \item Benchmark model-based \(\mathcal{H}_2\) synthesis assuming perfect model knowledge (\cite{Scherer-Weiland-2011-LMIs-in-Control}).
\end{enumerate}

All optimization problems are solved using the Mosek solver (\cite{mosek}) with the Yalmip toolbox (\cite{Lofberg-2004-YALMIP}) in MATLAB. Note that in the black-box case (i), \(\overline{d}_i\) must be scaled by \(\lVert B_{d,i} \rVert_2^2\).
We compare the closed-loop performance of controllers synthesized by each design method as the noise bound \(\overline{d}_i\) increases in Figure \ref{fig:noisy_results}, with a \(\log x\) axis. We use Theorem \ref{thm:distributed_synthesis_grey_box} for methods (i) and (ii).
\begin{figure}[h]
    \centering
    \includegraphics[width=0.85\linewidth]{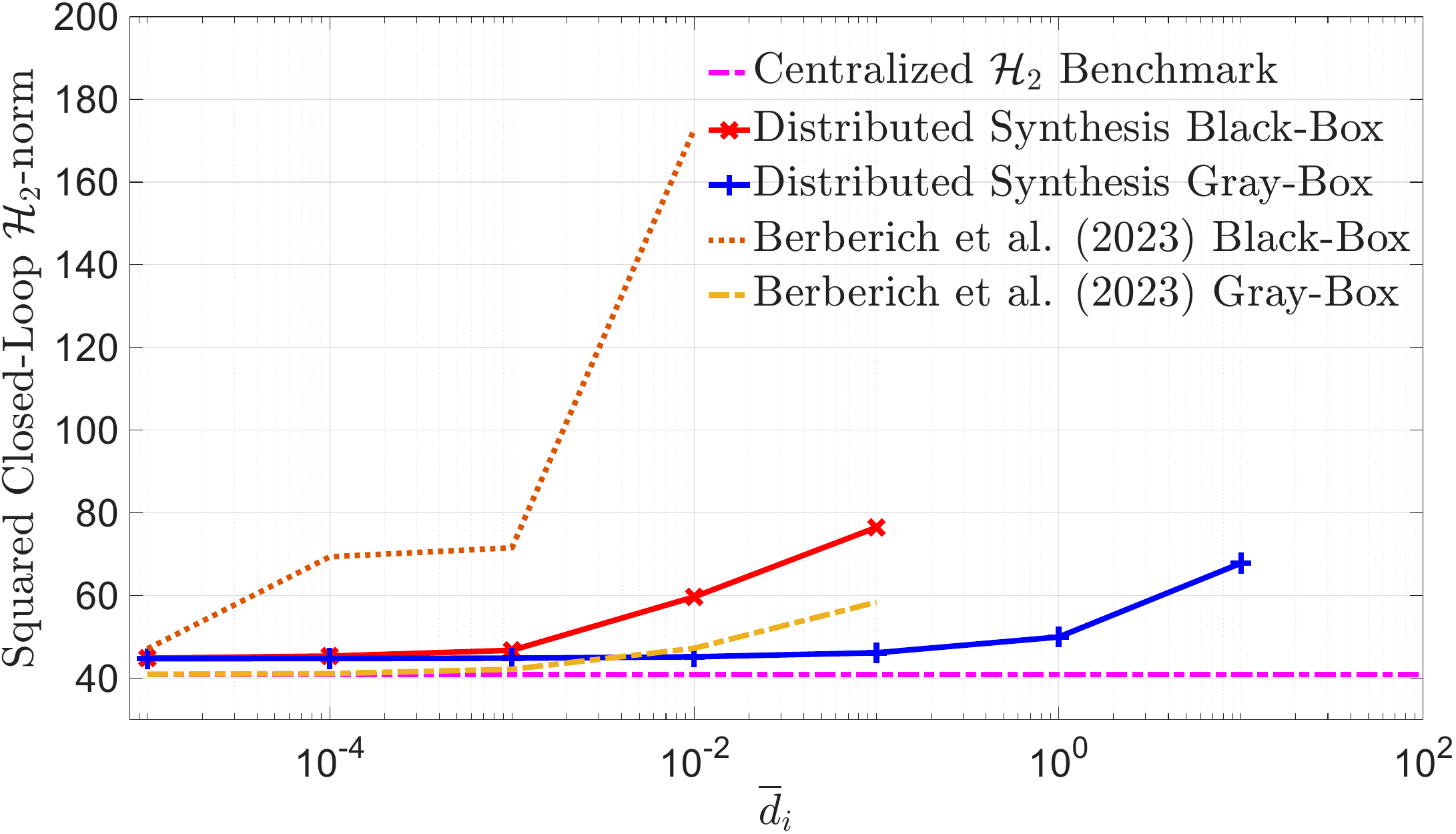}
    \caption{Squared closed-loop \(\mathcal{H}_2\)-norm under control design methods (i)-(v)}
    \label{fig:noisy_results}
\end{figure}
At very low noise levels, centralized gray-box (iv) achieves indistinguishable performance compared to the \(\mathcal{H}_2\) benchmark (v), whilst our distributed black-box (i) and gray-box (ii) schemes have around 10\% performance degradation due to the conservatism introduced to decouple the global Lyapunov matrix and constraint \eqref{eqn:relaxation_inequality_global} and the restriction of controller structure by \(\mathcal{G}_C\). Centralized black-box (iii) consistently has the worst performance over all noise levels and has the smallest upper bound on \(\overline{d}_i\) before infeasibility occurs, as no knowledge of the physical coupling topology or dynamics is included. We observe that (i) performs better, as the control design is inherently informed of the system structure.

Our distributed gray-box synthesis (ii) always outperforms its black-box counterpart (i) due to the inclusion of local model knowledge. Interestingly, (ii) starts to outperform (iv) as \(\overline{d}_i\) increases and remains feasible at a higher noise bound, which appears counterintuitive as (iv) provides an unstructured controller and avoids conservative conditions in synthesis. It can be explained by considering the structured way \(\overline{d}_i\) enters the distributed formulation (ii), whilst (iv) uses a single \(\overline{d} = \sum_{i\in\mathcal{V}} \overline{d}_i\). Perhaps introducing structured data-driven multipliers into the method in \cite{Berberich-et-al-2023-Prior-Knowledge-Data-Robust-Design} could reduce this gap. Finally, all data-driven methods become infeasible once the noise bound is sufficiently high.



Theorem \ref{thm:distributed_opt} guarantees that controllers synthesized using ADMM will asymptotically converge to those found by directly solving \eqref{eqn:distributed_synthesis_LMIs} as the number of ADMM iterations increases. The convergence in closed-loop control performance under a controller synthesized at a given number of iterations is shown in Figure \ref{fig:ADMM_conv} with a \(\log y\) axis, for \(\overline{d}_i = 10^{-5}\) under distributed gray-box synthesis (ii) and with \(\alpha = 2.5\). Similar convergence performance was observed across the domain of \(\overline{d}_i\) values tested.
\begin{figure}[h]
    \centering
    \includegraphics[width=0.85\linewidth]{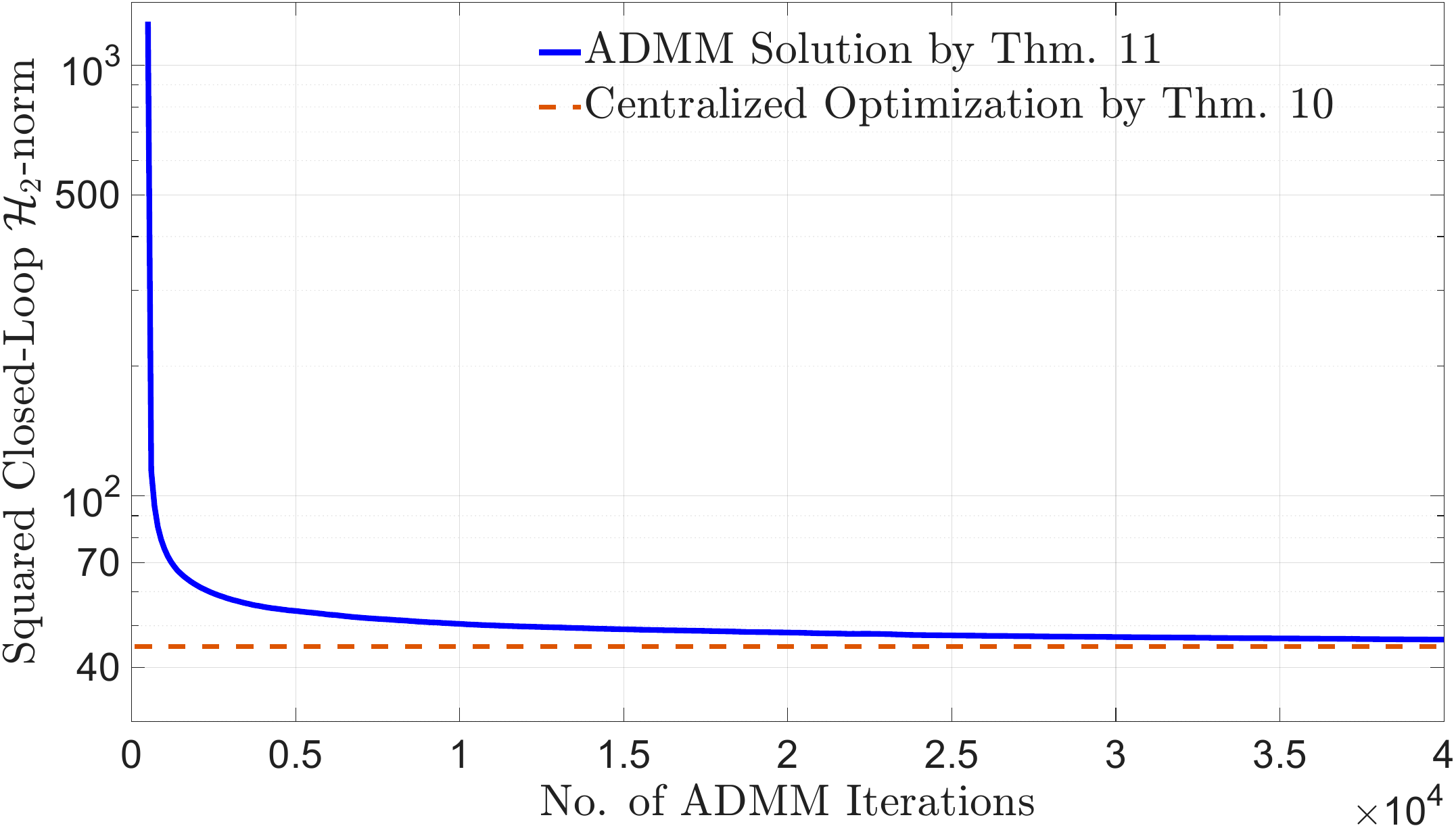}
    \caption{Convergence of control performance as the number of ADMM iterations increases}
    \label{fig:ADMM_conv}
\end{figure}

\section{Conclusion}
\label{section:conclusion}

In this paper, we proposed a distributed synthesis approach to the gray-box design of distributed \(\mathcal{H}_2\) controllers. The scheme is able to design controllers from data whilst incorporating available model knowledge and prior knowledge on the unknown dynamics. Avenues for future work include extending the method to the output-feedback setting and investigating the possibility of adaptive plug-and-play controller design.

\begin{ack}
The authors would like to thank Dr Matthias Köhler, Jiaxin Wang and Michael Cummins for insightful discussions whilst preparing this work.
\end{ack}


\bibliography{bibliography}
                                                   







\end{document}